\documentclass[12pt]{article}
\usepackage[margin=2.9cm]{geometry}
\usepackage{algorithmicx}

\usepackage[T1]{fontenc}
\usepackage[utf8]{inputenc}
\usepackage[english]{babel}
\usepackage{amsmath,amssymb, amsthm}
\usepackage[bookmarksopen,colorlinks,linkcolor=blue]{hyperref}
\usepackage{algorithm}
\usepackage{algpseudocode}
\usepackage{authblk}
\usepackage[shortlabels]{enumitem}
\usepackage{graphicx}
\usepackage{enumitem}

\usepackage{lmodern}
\usepackage{color}

\newtheorem{theorem}{Theorem}
\newtheorem{corollary}[theorem]{Corollary}
\newtheorem{lemma}[theorem]{Lemma}
\newtheorem{proposition}[theorem]{Proposition}
\newtheorem{definition}{Definition}

\newtheorem*{remark}{Remark}

\renewcommand{\le}{\leqslant}
\renewcommand{\ge}{\geqslant}

\makeatletter
\renewcommand*{\ALG@name}{Invariant}
\makeatother

\def\Pr{\mathop{\mathbf{Pr}}}
\def\MAJ{\mathrm{MAJ}}
\def\THR{\mathrm{THR}}
\author[1]{Alexander Kozachinskiy\thanks{Alexander.Kozachinskiy@warwick.ac.uk.  Supported by the EPSRC grant EP/P020992/1 (Solving Parity Games in Theory and Practice).}
}\author[2,3]{
Vladimir Podolskii\thanks{podolskii@mi-ras.ru}}

\affil[1]{University of Warwick, Coventry, UK}

\affil[2]{Steklov Mathematical Institute, Russian Academy of Sciences, Moscow, Russia}
\affil[3]{National Research University Higher School of Economics, Moscow, Russia}
\title{Multiparty Karchmer -- Wigderson Games and Threshold Circuits}

\date{}

\begin{document}

\maketitle

\begin{abstract}
We suggest a generalization of Karchmer -- Wigderson communication games to the multiparty setting. Our generalization turns out to be tightly connected to circuits consisting of threshold gates. This allows us to obtain new explicit constructions of such circuits for several functions.
In particular, we provide an explicit (polynomial-time computable) log-depth monotone formula for Majority function, consisting only of 3-bit majority gates and variables.
This resolves  a conjecture of Cohen et al. (CRYPTO 2013). 
\end{abstract}

\tableofcontents

\section{Introduction} \label{sec:intro}

Karchmer and Wigderson established tight connection between circuit depth and communication complexity~\cite{KW90} (see also \cite[Chapter 9]{rao_yehudayoff_2020}). Namely, they showed that for each Boolean function $f$ one can define a communication game which communication complexity \emph{exactly} equals the depth of $f$ in the standard De Morgan basis. This discovery turned out to be very influential in Complexity Theory. A lot of circuit depth lower bounds as well as formula size lower bounds rely on this discovery~\cite{KRW95, RM97, GMWW2014, GP2014, DM2018}. Karchmer -- Wigderson games have been used also in adjacent areas like Proof Complexity (see, e.g.,~\cite{Sok2017}).

Karchmer -- Wigderson games represent a deep connection of \emph{two-party} communication protocols with De Morgan circuits. Loosely speaking, in this connection one party is responsible for $\land$ gates and the other party is responsible for $\lor$ gates. In this paper we address the question of what would be a natural generalization of Karchmer -- Wigderson games to the multiparty setting. Is it possible to obtain in this way a connection with other types of circuits?

We answer positively to this question: we suggest such a generalization and show its connection to circuits consisting of \emph{threshold gates}. To motivate our results we first present applications we get from this new connection.

\subsection{Applications to circuits}

There are two classical constructions of  $O(\log n)$-depth monotone formulas for the Majority function,  $\MAJ_{2n + 1}$. The one was given by Valiant~\cite{val84}. Valiant used  probabilistic method which does not give an explicit construction. The other construction is the AKS sorting network~\cite{AKS83}. This construction actually  gives polynomial-time computable $O(\log n)$-depth $O(n\log n)$-size monotone circuit for $\MAJ_n$. 

Several authors (see, e.g.,~\cite{Gol2011,Cohen2013}) noticed that the Valiant's probabilistic argument actually gives a $O(\log n)$-depth formula for $\MAJ_n$, consisting only of $\MAJ_3$ gates and variables. 
 Is it possible to construct a $O(\log n)$-depth circuit for $\MAJ_{2n + 1}$, consisting only of $\MAJ_3$ gates and variables, \emph{deterministically in polynomial time}?\footnote{Note that AKS sorting network does not provide a solution  because it consists of $\land$ and $\lor$ gates.}
 
This question was stated as a conjecture by Cohen et al. in~\cite{Cohen2013}. First, they showed that the answer is positive under some cryptographic assumptions. Secondly, they constructed (unconditionally) a polynomial-time computable $O(\log n)$-depth circuit, consisting only of $\MAJ_3$ gates and variables, which coincides with $\MAJ_n$ for all inputs in which the fraction of ones is bounded away from $1/2$ by $2^{-\Theta(\sqrt{\log n})}$.

We show that the conjecture of Cohen et al. is true (unconditionally).

\begin{theorem}
\label{cohen_conj_1}
There exists polynomial-time computable $O(\log n)$-depth formula for $\MAJ_{2n + 1}$, consisting only of $\MAJ_3$ gates and variables.
\end{theorem}
In the proof we use the AKS sorting network. In fact, one can use any construction of polynomial-time computable $O(\log n)$-depth monotone circuit for $\MAJ_{2n + 1}$. We also obtain the following general result:

\begin{theorem}
\label{transformation}
If there is a monotone formula (i.e., formula, consisting of $\land, \lor$ gates and variables) for $\MAJ_{2n + 1}$ of size $s$, then there is a formula for $\MAJ_{2n + 1}$ of size $O(s \cdot n^{\log_2(3)}) = O(s \cdot n^{1.58\ldots})$, consisting only of $\MAJ_3$ gates and variables.
\end{theorem}
Transformation from the last theorem, however, is not efficient. We can make this transformation polynomial-time computable, provided $\log_2(3)$ is replaced by $1/(1 - \log_3(2)) \approx 2.71$.  In turn, we view Theorem \ref{transformation}
as a potential approach to obtain super-quadratic lower bounds on monotone formula size for $\MAJ_{2n + 1}$. However, this approach requires better than $n^{2 +\log_2(3)}$ lower bound on formula size of $\MAJ_{2n + 1}$ in the $\{\MAJ_3\}$ basis. Arguably, this basis may be easier to analyze than the standard monotone basis. The best known size upper bounds  in the $\{\land, \lor\}$ basis and the $\{\MAJ_3\}$ basis are, respectively, $O(n^{5.3})$ and $O(n^{4.29})$~\cite{GM96}. Both bounds are due to Valiant's method (see~\cite{GM96} also for the limitations of Valiant's method). 

We also study a generalization of the conjecture of Cohen et al. to threshold functions. By $\THR^b_a$ we denote the following Boolean function:
$$\THR^b_a\colon\{0, 1\}^b \to \{0, 1\}, \qquad \THR^b_a(x) = \begin{cases} 1 & \mbox{$x$ contains at least $a$ ones,} \\ 0 &\mbox{otherwise.} \end{cases}$$
For some reasons (to be discussed below) a natural generalization would be a question of whether $\THR^{kn + 1}_{n + 1}$ can by computed by a $O(\log n)$-depth circuit, consisting only of $\THR^{k+1}_2$ gates and variables (initial conjecture can be obtained by setting $k = 2$). This question was also addressed by Cohen et al. in~\cite{Cohen2013}. First, they observed that there is a construction of depth $O(n)$ (and exponential size). Secondly, they gave an explicit construction of depth $O(\log n)$, which coincides with $\THR^{kn + 1}_{n + 1}$ for all inputs in which the fraction of ones is bounded away from $1/k$ by $\Theta(1/\sqrt{\log n})$.

However, no exact (even non-explicit) construction with sub-linear depth or sub-exponential size was known. In particular, Valiant's probabilistic construction does not work for $k\ge 3$.
Nevertheless, in this paper we improve depth $O(n)$ to $O(\log^2 n)$ and size from $\exp\{O(n)\}$ to $n^{O(1)}$ for this problem:

\begin{theorem}
\label{cohen_2}
For any constant $k\ge 3$ 
there exists polynomial-time computable $O(\log^2 n)$-depth polynomial-size circuit for $\THR^{k n + 1}_{n + 1}$, consisting only of $\THR^{k + 1}_2$ gates and variables.
\end{theorem}

\subsection{Applications to Multiparty Secure Computations}

The conjecture stated in~\cite{Cohen2013} was motivated by applications to Secure Multiparty Computations. The paper~\cite{Cohen2013} establishes an approach to construct efficient multiparty protocols based on protocols for small number of players. More specifically, in their framework one starts with a protocol for small number of players and a formula $F$ computing certain boolean function. Then one combines a protocol for a small number of players with itself recursively, where the recursion mimics the formula $F$.

It is shown in~\cite{Cohen2013} that from our result it follows that for any $n$ there is an explicit polynomial size protocol for $n$ players secure against a passive adversary that controls any $t < \frac n2$ players. It is also implicit in~\cite{Cohen2013} that from Theorem~\ref{cohen_2} for $k=3$ it follows that for any $n$ there is a protocol of size $2^{O(\log^2 n)}$ for $n$ players secure against an active adversary that controls any $t < \frac n3$ players. An improvement of the depth of the formula in Theorem~\ref{cohen_2} to $O(\log n)$ would result in a polynomial size protocol. We refer to~\cite{Cohen2013} for more details on the secure multiparty computations.

\subsection{Multiparty Karchmer -- Wigderson games}

We now reveal a bigger picture to which the above results belong to. Namely, they can be put into framework of multiparty Karchmer -- Wigderson games.

Before specifying how we define these games let us give an instructive example. Consider ordinary monotone Karchmer -- Wigderson game for $\MAJ_{2n + 1}$. In this game Alice receives a string $x\in \MAJ_{2n + 1}^{-1}(0)$ and Bob receives a string $y \in\MAJ_{2n + 1}^{-1}(1)$. In other words, the number of ones in $x$ is at most $n$ and the number of ones in $y$ is at least  $n + 1$. The goal of Alice and Bob is to find some coordinate $i$ such that $x_i = 0$ and $y_i = 1$. Next, imagine that Bob flips each of his input bits. After that parties have two vectors in both of which the number of ones is at most $n$. Now Alice and Bob have  to find any coordinate in which both vectors are $0$.

In this form this problem can be naturally generalized to the multiparty setting. Namely, assume that there are $k$ parties, and each receives a Boolean vector of length $kn + 1$ with at most $n$ ones. Let the task of parties be to find a coordinate in which all $k$ input vectors are $0$. How many bits of communication are needed for that?

For $k = 2$ the answer is $O(\log n)$, because there exists a $O(\log n)$-depth monotone circuit for $\MAJ_{2n + 1}$ and hence the monotone Karchmer -- Wigderson game for $\MAJ_{2n + 1}$ can be solved in $O(\log n)$ bits of communication. For $k \ge 3$ we are only aware of a simple $O(\log^2 n)$-bit solution based on the binary search.

Now, let us look at the case $k\ge 3$ from another perspective and introduce multiparty Karchmer -- Wigderson games.
Note that each party receives a vector on which $\THR^{kn + 1}_{n + 1}$ equals $0$. The goal is to find a common zero. Note that we can consider a similar problem for any function $f$ satisfying so-called \emph{$Q_k$-property}: any $k$ vectors from $f^{-1}(0)$ have a common zero. In the next definition we define $Q_k$-property formally and also introduce related $R_k$-property.

\begin{definition}
Let $Q_k$ be the set of all Boolean functions $f$ satisfying the following property: for all $x^1, x^2, \ldots, x^k\in f^{-1}(0)$ there is a coordinate $i$ such that $x^1_i = x^2_i = \ldots =  x^k_i = 0$.

Further, let $R_k$ be the set of all Boolean functions $f$  satisfying the following property: for all $x^1, x^2, \ldots, x^k\in f^{-1}(0)$ there is a coordinate $i$ such that $x^1_i = x^2_i = \ldots =  x^k_i$.
\end{definition}


For $f\in Q_k$ let \emph{$Q_k$-communication game} for $f$  be the following communication problem. In this problem there are $k$ parties. The $j$th party receives a Boolean vector $x^j\in f^{-1}(0)$. The goal of players is to find any coordinate $i$ such that $x^1_i = x^2_i = \ldots = x^k_i = 0$. 

Similarly we can define \emph{$R_k$-communication games} for functions from $R_k$. In the $R_k$-communication games the objective of parties is slightly different: their goal is to find any coordinate $i$  and a bit $b$ such that $x^1_i = x^2_i = \ldots = x^k_i = b$.

Self-dual functions belong to $R_2$ and monotone self-dual functions belong to $Q_2$. It is easy to see that $R_2$-communication games are equivalent to Karchmer -- Wigderson games for self-dual functions (one party should flip all the input bits). Moreover, $Q_2$-communication games are equivalent to monotone Karchmer -- Widgerson games for monotone self-dual functions. 

In this paper we consider $R_k$-communication games as a multiparty generalization   of Karchmer -- Wigderson games. In turn, $Q_k$-communication games are considered as a generalization of \emph{monotone} Karchmer -- Wigderson games. To justify this choice one should relate them to some type of circuit complexity. 

\subsection{Connection to threshold gates and the main result}

Every function from $Q_k$ can be \emph{lower bounded} by a circuit, consisting only of $\THR^{k + 1}_2$ gates and variables.  More precisely, let us write $C \le f$ for a Boolean circuit $C$ and a Boolean function $f$ if for all $x\in f^{-1}(0)$ we have $C(x) = 0$. Then the following proposition holds:
\begin{proposition}[\cite{Cohen2013}]
\label{Q_k_char}
The set $Q_k$ is equal to the set of all Boolean functions $f$ for which there exists a circuit $C\le f$, consisting only of $\THR^{k + 1}_2$ gates and variables.
\end{proposition}

There is a similar characterization of the set $R_k$.
\begin{proposition}
\label{R_k_char}
The set $R_k$ is equal to the set of all Boolean functions $f$ for which there exists a circuit $C \le f$, consisting only of $\THR^{k + 1}_2$ gates and literals\footnote{We stress that negations can only be  applied to variables but not to $\THR^{k + 1}_2$ gates.}.
\end{proposition}
The proof from~\cite{Cohen2013} of Proposition \ref{Q_k_char} with obvious modifications also works for Proposition \ref{R_k_char}.

Given $f\in Q_k$, what is the minimal depth of a circuit $C\le f$,  consisting only of $\THR^{k + 1}_2$ gates and variables? We show that this quantity is equal (up to constant factors) the communication complexity of $Q_k$-communication game for $f$.

\begin{theorem}
\label{Q_k_main}
Let $k \ge 2$ be any constant. Then for any $f\in Q_k$ the following two quantities are equal up to constant factors:
\begin{itemize}
\item the communication complexity of $Q_k$-communication game for $f$;
\item minimal $d$ for which there exists a $d$-depth circuit $C \le f$, consisting only of $\THR^{k + 1}_2$ gates and variables.
\end{itemize}
\end{theorem}

 Similar result can be obtained for $R_k$-communication games.

\begin{theorem}
\label{R_k_main}
Let $k \ge 2$ be any constant. Then for any $f\in R_k$ the following two quantities are equal up to constant factors:
\begin{itemize}
\item the communication complexity of $R_k$-communication game for $f$;
\item minimal $d$ for which there exists a $d$-depth circuit $C \le f$, consisting only of $\THR^{k + 1}_2$ gates and literals.
\end{itemize}
\end{theorem}

Proofs of both theorems are divided into two parts:
\begin{enumerate}[label=(\alph*)]
\item transformation of a $d$-depth  circuit $C \le f$, consisting only of $\THR^{k + 1}_2$ gates and variables (literals), into a $O(d)$-bit protocol computing $Q_k$($R_k$)-communication game for $f$;
\item transformation of a $d$-bit protocol computing $Q_k$($R_k$)-communication game for $f$ into a  $d$-depth  circuit $C\le f$, consisting only of $\THR^{k + 1}_2$ gates and variables (literals).
\end{enumerate}

The first part is simple and the main challenge is the second part. Later in this paper (Section 6) we also formulate refined versions of Theorems \ref{Q_k_main} and \ref{R_k_main}. Namely, we refine these theorems in the following two directions. Firstly, we take into account circuit size and for this we consider dag-like communication protocols. Secondly, we show that transformations (a-b) can be done in polynomial time (under some mild assumptions). 

 We derive our  upper bounds on the depth of $\MAJ_{2n + 1}$ and $\THR^{kn + 1}_{n + 1}$ (Theorems \ref{cohen_conj_1} and \ref{cohen_2}) from Theorem \ref{Q_k_main}. We first solve the corresponding $Q_k$-communication games with small number of bits of communication. Namely, for the case of $\MAJ_{2n + 1}$ we use AKS sorting network to solve the corresponding $Q_2$-communication game with $O(\log n)$ bits of communication. For the case of $\THR^{kn + 1}_{n + 1}$ with $k\ge 3$ we solve the corresponding $Q_k$-communication game by a simple binary search protocol with $O(\log^2 n)$ bits of communication. This is where we get depth $O(\log n)$ for Theorem \ref{cohen_conj_1} and depth $O(\log^2 n)$ for Theorem \ref{cohen_2}. Again, some special measures should be taken to make the resulting circuits polynomial-time computable and to control their size\footnote{We should only care about the size in case of Theorem \ref{cohen_2}, because depth $O(\log n)$ immediately gives polynomial size.}.

\subsection{Our techniques: $Q_k$($R_k$)-hypotheses games}

As we already mentioned, the hard part of our main result is to transform a protocol into a circuit.

For this we develop a new language to describe circuits, consisting of threshold gates. Namely, for every $f$ in $Q_k$ ($R_k$) we introduce the corresponding \emph{$Q_k$($R_k$)-hypotheses game} for $f$. We show that strategies in these games exactly capture depth and size of circuits, consisting only of $\THR^{k + 1}_2$ gates and variables (literals). 
It turns out that strategies are more convenient than circuits to simulate protocols, since they operate in the same top-bottom manner.

Once we establish the equivalence of circuits and hypotheses games, it remains for us to transform a communication protocol into a strategy in a hypotheses games. This is an elaborate construction that is presented in Propositions~\ref{Q_k_to_circuits} and~\ref{R_k_to_circuits}. Below in this section we introduce hypotheses games and as an illustration sketch the construction of a strategy in a hypothesis game that is used in the proof of Theorem~\ref{cohen_conj_1}.

Here is how we define these games. Fix $f\colon\{0, 1\}^n \to \{0, 1\}$. There are two players, Nature and Learner. Before the game starts, Nature privately chooses $z\in f^{-1}(0)$, which is then can not be changed. The goal of Learner is to find some $i\in [n]$ such that $z_i = 0$. The game proceeds in rounds. At each round Learner specifies $k + 1$ families $\mathcal{H}_0, \mathcal{H}_1, \ldots, \mathcal{H}_{k} \subset f^{-1}(0)$ to Nature. We understand this as if Learner makes the following $k + 1$ hypotheses about $z$:
\begin{align*}
``z&\in\mathcal{H}_0\mbox{''}, \\
``z&\in\mathcal{H}_1\mbox{''}, \\
&\vdots \\
``z&\in\mathcal{H}_{k}\mbox{''}.
\end{align*}
Learner looses immediately if less than $k$ hypotheses are true, i.e., if the number of $j\in\{0, 1, \ldots, k\}$ satisfying $z\in\mathcal{H}_j$ is less than $k$. Otherwise Nature points out to some hypothesis which is true. In other words, Nature specifies to Learner some $j\in \{0, 1, \ldots, k\}$ such that $z\in \mathcal{H}_j$. The game then proceeds in the same manner for some finite number of rounds. At the end Learner outputs an integer $i\in [n]$. We say that Learner wins if  $z_i = 0$.

It is not hard to show that Learner has a winning strategy in $Q_k$-hypotheses game for $f$ if and only if $f\in Q_k$. Since we will use similar arguments in the paper, let us go through the ``if'' part: if $f\in Q_k$, then Learner has a winning strategy. Denote by $\mathcal{Z}$ be the set of all $z$'s which are compatible with Nature's answers so far. At the beginning $\mathcal{Z} = f^{-1}(0)$. If $|\mathcal{Z}| \ge k + 1$, Learner takes any distinct $z^1, z^2, \ldots, z^{k + 1} \in \mathcal{Z}$ and makes the following hypotheses:
\begin{align*}
``z&\neq z^1\mbox{''}, \\
``z&\neq z^2\mbox{''}, \\
&\vdots \\
``z&\neq z^{k + 1}\mbox{''}.
\end{align*}
At least $k$ hypotheses are true, and the Nature's response strictly reduces the size of $\mathcal{Z}$. When the size of $\mathcal{Z}$ becomes $k$, Learner is ready to give an answer due to $Q_k$-property of $f$.

This strategy requires exponential in $n$ number of rounds. This can be easily improved to $O(n)$ rounds. Indeed, instead of choosing $k + 1$ distinct elements of $\mathcal{Z}$ split $\mathcal{Z}$ into $k + 1$ disjoint almost equal parts. Then let the $i$th hypotheses be ``$z$ is not in the $i$th part''. Nature's response to this reduces the size of $\mathcal{Z}$ by a constant factor, until the size of $\mathcal{Z}$ is $k$.

For $f\in Q_k$ we can now ask what is the minimal number of rounds on in a Learner's winning strategy. The following proposition gives an exact answer:
\begin{proposition}
\label{games_Q_k}
For any $f\in Q_k$ the following holds. Learner has a $d$-round winning strategy in $Q_k$-hypotheses game for $f$ if and only if there exists a $d$-depth circuit $C\le f$, consisting only of $\THR^{k+1}_2$ gates and variables. 
\end{proposition}

Proposition \ref{games_Q_k} is the core result for our applications. For instance, we prove Theorem \ref{cohen_conj_1} by giving an explicit $O(\log n)$-round winning strategy of Learner in $Q_2$-hypotheses game for $\MAJ_{2n + 1}$. Let us now sketch our argument (the complete proof can be found in Section~\ref{sec:majority}).

Assume that Nature's input vector is $z$. We notice that in $O(\log n)$ rounds one can easily find \emph{two} integers $i, j \in [2n + 1]$ such that either $z_i = 0$ or $z_j = 0$. However, we need to know for sure. For that we take any polynomial time computable $O(\log n)$-depth monotone formula $F$ for $\MAJ_{2n + 1}$ (for instance one that can be obtained from the AKS sorting network). We start to descend  from the output gate of $F$ to one of $F$'s inputs. Throughout this descending we maintain the following invariant. If $g$ is the current gate, then either $g(z) = 0 \land z_i = 0$ or $g(\lnot z) = 1 \land z_j = 0$ (here $\lnot$ denotes bit-wise negation). It can be shown that in one round one can either exclude $i$ or $j$ (which will already give us an answer) or replace $g$ by some gate which is fed to $g$. If we reach an input to $F$, we output the index of the corresponding variable.

Similarly one can define $R_k$-hypotheses game for any $f\colon\{0, 1\}^n \to\{0, 1\}$. In $R_k$-hypotheses game Nature and Learner play in the same way except that now Learner's objective is to find some pair $(i, b) \in [n] \times \{0, 1\}$ such that $z_i =b$.  The following analog of Proposition \ref{games_Q_k} holds:
\begin{proposition}
\label{games_R_k}
For any $f\in R_k$ the following holds. Learner has a $d$-round winning strategy in $R_k$-hypotheses game for $f$ if and only if there exists a $d$-depth circuit $C\le f$, consisting only of $\THR^{k+1}_2$ gates and literals. 
\end{proposition}

\subsection{Organization of the paper}

In Section~\ref{sec:prelim} we give Preliminaries. In Section~\ref{sec:hypotheses_games} we define $Q_k$($R_k$)-hypotheses games formally and derive Proposition \ref{games_Q_k} and \ref{games_R_k}. In Section~\ref{sec:majority} we obtain our results for Majority function (Theorems \ref{cohen_conj_1} and \ref{transformation}) using  simpler arguments than in our general results. Then in Section~\ref{sec:proof_main} we prove these general results (Theorems \ref{Q_k_main} and \ref{R_k_main}). In Section~\ref{sec:effective} we refine Theorems  \ref{Q_k_main} and \ref{R_k_main} in order to take into account the circuit size and computational aspects (Theorems \ref{efficient_Q_k} and \ref{efficient_R_k} below). In Section~\ref{sec:corol} we derive Theorem \ref{cohen_2} and provide another proof for Theorem \ref{cohen_conj_1}. Finally, in Section~\ref{sec:open_prob} we formulate some open problems.

\section{Preliminaries} \label{sec:prelim}
Let $[n]$ denote the set $\{1, 2, \ldots, n\}$ for $n\in\mathbb{N}$. For a set $W$ we denote the set of all subsets of $W$ by $2^W$. For two sets $A$ and $B$ by $A^B$ we mean the set of all functions of the form $f\colon B\to A$.

We usually use subscripts to denote coordinates of vectors. In turn, we usually use superscripts to numerate vectors.

 We use standard terminology for Boolean formulas and circuits~\cite{Jukna}. We denote the size of a circuit $C$ by $\mathrm{size}(C)$ and the depth by $\mathrm{depth}(C)$. By De Morgan formulas/circuits we mean formulas/circuits consisting of $\land, \lor$ gates of fan-in 2 and literals (i.e., we assume that negations are applied only to variables). By monotone formulas/circuits we mean formulas/circuits consisting of $\land, \lor$ gates of fan-in 2 and variables. We also consider formulas/circuits consisting only of $\THR^{k + 1}_2$ gates and variables (literals). We stress that in such circuits we do not use constants. Allowing literals as inputs we allow to apply negations only to variables. We also assume that negations in literals do not contribute to the depth of a circuit.

We use the notion of deterministic communication protocols in the multiparty \emph{number-in-hand} model. However, to capture the circuit size in our results we consider not only standard \emph{tree-like} protocols, but also \emph{dag-like} protocols. This notion was considered by Sokolov in~\cite{Sok2017}. We use   slightly different variant of this notion, arguably more intuitive one. In the next subsection we provide all necessary definitions.
To obtain a definition of a standard protocol  one should replace dags by binary trees.

\subsection{Dags and dag-like communication protocols}

We use the following terminology for directed acyclic graphs (dags). Firstly, we allow more than one directed edge from one node to another.
A terminal node of a dag $G$ is a node with no out-going edges. Given a dag $G$, let
\begin{itemize}
\item $V(G)$ denote the set of nodes of $G$;
\item $T(G)$ denote the  set of terminal nodes of $G$. 
\end{itemize}
For $v\in V(G)$ let $Out_G(v)$ be the set of all edges of $G$ that start at $v$.
A dag $G$ is called $t$-ary if every non-terminal node $v$ of $G$ we have $|Out_G(v)| = t$. An ordered $t$-ary dag is a $t$-ary dag $G$ equipped with a mapping from the set of edges of $G$ to $\{0, 1, \ldots, t - 1\}$. This mapping restricted to $Out_G(v)$ should be injective for every $v\in V(G)\setminus T(G)$. The value of this mapping on an edge $e$ will be called the \emph{label} of $e$. In terms of labels we require for ordered $t$-ary dags that any $t$ edges, starting at the same node, have different labels.

By a path in $G$ we mean a sequence of \emph{edges} $\langle e_1, e_2, \ldots, e_m\rangle$ such that for every $j\in[m - 1]$ edge $e_i$ ends in the same node in which $e_{j + 1}$ starts.  Note that there may be two distinct paths visiting same nodes (for instance, there may be two parallel edges from one node to another).

We say that a node $w$ is a descendant of a node $v$ if there is a path from $v$ to $w$. We call $w$ a successor of $v$ if there is an edge from $v$ to $w$. A node $s$ is called \emph{starting node} if any other node is a descendant of $s$. Note that any dag has at most one starting node. 

If a dag $G$ has the starting node $s$, then by depth of $v\in V(G)$ we mean the maximal length of a path from $s$ to $v$. The depth of $G$ then is the maximal depth of its nodes. 

Assume that  $\mathcal{X}_1, \mathcal{X}_2, \ldots, \mathcal{X}_k, \mathcal{Y}$ are some finite sets.
\begin{definition}
\label{dag_like_protocol}
A $k$-party dag-like communication protocol $\pi$ with inputs from $\mathcal{X}_1 \times \mathcal{X}_2 \times \ldots \mathcal{X}_k$ and with outputs from $\mathcal{Y}$ is a tuple $\langle G, P_1, P_2, \ldots, P_k, \phi_1, \phi_2, \ldots, \phi_k, l\rangle$, where
\begin{itemize}
\item $G$ is an ordered $2$-ary dag with the starting node $s$;
\item $P_1, P_2, \ldots, P_k$ is a partition of $V(G)\setminus T(G)$ into $k$ disjoint subsets;
\item $\phi_i$ is a function from $P_i\times \mathcal{X}_i$ to $\{0, 1\}$;
\item $l$ is a function from $T(G)$ to $\mathcal{Y}$.
\end{itemize}
\end{definition}

The depth of $\pi$ (denoted by $\mathrm{depth}(\pi)$) is the depth of $G$. The size of $\pi$ (denoted by $\mathrm{size}(\pi)$) is  $|V(G)|$.

 The underlying mechanics of the protocol is as follows. Parties descend from $s$ to one of the terminals of $G$. If the current node $v$ is not a terminal and $v\in P_i$, then at $v$ the $i$th party communicates a bit to all the other parties. Namely, the $i$th party communicates the bit $b = \phi_i(v, x)$, where $x\in\mathcal{X}_i$ is the input of the $i$th party. Among the two edges, starting at $v$, parties choose one labeled by $b$ and descend to one of the successors of $v$ along this edge. Finally, when parties reach a terminal $t$, they output $l(t)$. 

We say that $x\in\mathcal{X}_i$ is $i$-compatible with an edge $e$ from $v$ to $w$ if one of the following two condition holds:
\begin{itemize}
\item $v\notin P_i$;
\item $v\in P_i$ and $e$ is labeled by $\phi_i(v, x)$.
\end{itemize}
We say that $x\in \mathcal{X}_i$ is $i$-compatible with  a path $p = (e_1, e_2, \ldots, e_m)$ of $G$ if for every $j\in [m]$ it holds that $x$ is $i$-compatible with $e_j$. Finally, we say that $x\in\mathcal{X}_i$ is $i$-compatible with a node $v\in V(G)$ if there is a path $p$ from $s$ to $v$ such that $x$ is $i$-compatible with $v$.

We say that an input  $(x^1, x^2, \ldots, x^k) \in \mathcal{X}_1 \times \mathcal{X}_2 \times \ldots \mathcal{X}_k$ visits a node $v\in V(G)$ if there is a path $p$ from $s$ to $v$ such that for every $i\in [k]$ it holds that $x^i$ is $i$-compatible with $p$. Note that there is  unique $t\in T(G)$ such that $(x^1, x^2, \ldots, x^k)$ visits $t$.

To formulate an effective version of Theorems \ref{Q_k_main} and Theorem \ref{R_k_main} we need the following definition.
\begin{definition}
\label{light_form}
The \emph{light form} of a $k$-party dag-like communication protocol $\pi = \langle G, P_1, P_2, \ldots, P_k, \phi_1, \phi_2, \ldots, \phi_k, l\rangle$ is a tuple $\langle G, P_1, P_2, \ldots, P_k, l\rangle$.
\end{definition}
 I.e., to obtain the light form of $\pi$ we just forget about $\phi_1, \phi_2, \ldots, \phi_k$. In other words, the light form only contains  the underlying graph of $\pi$, the partition of non-terminal nodes between parties and the labels of terminals. On the other hand, in the light form there is no information at all how parties communicate at the non-terminal nodes.

Protocol $\pi$ \emph{computes} a relation $S\subset \mathcal{X}_1\times\mathcal{X}_2\times\ldots\times \mathcal{X}_k\times\mathcal{Y}$ if the following holds. For every $(x^1, x^2, \ldots, x^k) \in  \mathcal{X}_1\times\mathcal{X}_2\times\ldots\times \mathcal{X}_k$ there exist $y\in \mathcal{Y}$ and $t\in T(G)$ such that $(x^1, \ldots, x^k)$ visits $t$, $l(t) = y$ and $(x^1, x^2, \ldots, x^k, y) \in S$.

Using language of relations, we can formally define $Q_k$- and $R_k$-communication games. Namely, given $f\colon\{0, 1\}^n\to\{0,1\}, f\in Q_k$, we define $Q_k$-communication game for $f$ as the following relation:
\begin{align*}
S &\subset \underbrace{f^{-1}(0) \times \ldots \times f^{-1}(0)}_k \times [n], \\
 S &= \left\{(x^1, \ldots, x^k, j) \mid x^1_j  = \ldots = x^k_j = 0\right\}.
\end{align*}
Similarly, given $f\colon\{0, 1\}^n\to\{0,1\}, f\in R_k$, we define $R_k$-communication game for $f$ as the following relation:
\begin{align*}
S &\subset \underbrace{f^{-1}(0) \times \ldots \times f^{-1}(0)}_k \times ([n]\times\{0,1\}), \\
 S &= \left\{(x^1, \ldots, x^k, (j, b)) \mid x^1_j  = \ldots = x^k_j = b\right\}.
\end{align*}

It is easy to see that a dag-like protocol for $S$ can be transformed into a tree-like protocol of the same depth, but this transformation can drastically increase the size.

\section{Formal treatment of $Q_k$($R_k$)-hypotheses games} \label{sec:hypotheses_games}

Fix $f\in Q_k$, $f\colon\{0, 1\}^n\to\{0,1\}$. Here we define Learner's strategies in $Q_k$-hypotheses game for $f$ formally. We consider not only tree-like strategies but also dag-like. To specify a Learner's strategy $S$ in $Q_k$-hypotheses game we have to specify:
\begin{itemize}
\item An ordered $(k+1)$-ary dag $G$ with the starting node $s$;
\item a subset $\mathcal{H}_j(p)$ for every $j\in\{0, 1, \ldots, k\}$ and for every path $p$ in $G$ from $s$ to some node in $V(G)\setminus T(G)$;
\item a number $i_t\in [n]$ for every terminal $t$.
\end{itemize}
The underlying mechanics of the game is as follows. Let Nature's vector be $z\in f^{-1}(0)$. Learner and Nature descend from $s$ to one of the terminals of $G$. More precisely, a position in the game is determined by a path $p$, starting at $s$. If the endpoint of $p$ is not a terminal, then Learner specifies some sets $\mathcal{H}_0(p), \mathcal{H}_1(p), \ldots, \mathcal{H}_k(p)$ as his hypotheses. If less than $k$ of these sets contain $z$, then Nature wins. Otherwise Nature specifies some $j\in\{0, 1, \ldots, k\}$ such that $z\in\mathcal{H}_j(p)$. Among $k+1$ edges that start at the endpoint of $p$ players choose one which is labeled by $j$. After that they extend $p$ by this edge. At some point parties reach some terminal $t$ (i.e., the endpoint of $p$ becomes equal $t$). Then the game ends and Learner output $i_t$.

We stress that Learner's output depends only on $t$ but not on a path to $t$ (unlike Learner's hypotheses). This property will be crucial in establishing connection of $Q_k$-hypotheses games to circuits.

We now proceed to a formal definition of what does it mean that $S$ is winning for Learner.

We say that $z\in f^{-1}(0)$ is \emph{compatible} with a path $p = \langle e_1, \ldots, e_m\rangle$, starting in $s$, if the following holds. If $p$ is of length $0$, then every $z\in f^{-1}(0)$ is compatible with $p$. Otherwise for every $i\in\{1, \ldots, e_m\}$ it should hold that $z\in\mathcal{H}_j(\langle e_1, \ldots, e_{i - 1}\rangle)$, where $j$ is the label of edge $e_i$. Informally this means that Nature, having $z$ on input, can reach a position in the game which corresponds to a path $p$.

We say that strategy $S$ is winning for Learner in $Q_k$-hypotheses game for $f$ if for every path $p$, starting at $s$, and for every $z\in f^{-1}(0)$, compatible with $p$, the following holds:
\begin{itemize}
\item if the endpoint of $p$ is not a terminal, then the number of $j\in\{0, 1, \ldots, k\}$ such that $z\in \mathcal{H}_j(p)$ is at least $k$;
\item if the endpoint of $p$ is $t\in T(G)$, then $z_{i_t} = 0$.
\end{itemize}
We will formulate a stronger version of Proposition \ref{games_Q_k}. For that we need the notion of the \emph{light form} of the strategy $S$. Namely, the light form of $S$ is its underlying dag $G$ equipped with a mapping which to every $t\in T(G)$ assigns $i_t$. In other words, the light form contains a ``skeleton'' of $S$ and Learner's outputs in terminals (and no information about Learner's hypotheses). 

We can identify the light form of any strategy $S$ with a circuit, consisting only of $\THR^{k+1}_2$ gates and variables. Namely, place $\THR^{k + 1}_2$ gate in every $v\in V(G)\setminus T(G)$ and for every $t\in T(G)$ place a variable $x_{i_t}$ in $t$. Set $s$ to be the output gate. 

\begin{proposition}
\label{Q_k_equivalence}
For all $f\in Q_k, f\colon\{0,1\}^n\to\{0,1\}$ the following holds:
\begin{enumerate}[label=\textbf{(\alph*)}]
\item if $S$ if a Learner's winning strategy in $Q_k$-hypotheses game for $f$, then its light form, considered as a circuit $C$ consisting only of $\THR^{k + 1}_2$ gates and variables, satisfies $C\le f$.

\item  Assume that $C\le f$ is a circuit, consisting only of $\THR^{k + 1}_2$ gates and variables. Then there exists a Learner's winning strategy $S$ in $Q_k$-hypotheses game for $f$ such that the light form of $S$ coincides with $C$.
\end{enumerate}
\end{proposition}

We omit the proof of \textbf{\emph{(b)}} as in the paper we only use \textbf{\emph{(a)}}.

\begin{proof}[Proof of \textbf{\emph{(a)}} of Proposition \ref{Q_k_equivalence}]
For a node $v \in V(G)$ let $f_v\colon\{0, 1\}^n \to\{0, 1\}$ be the function, computed by the circuit $C$ at the gate, corresponding to $v$.

We shall prove the following statement. For any path $p$, starting in $s$, and for any $z$ which is compatible with $p$ it holds that $f_v(z) = 0$, where $v$ is the endpoint of $p$. To see why this implies $C\le f$ take any $z\in f^{-1}(0)$ and note that $z$ is compatible with the path of length $0$. The endpoint of such path is $s$ and hence $0 = f_s(z) = C(z)$.

We will prove the above statement by the backward induction on the length of $p$. The longest path $p$ ends in some $t\in T(G)$. By definition $f_t = x_{i_t}$. On the other hand, since $S$ is winning, $z_{i_t} = 0$ for any $z$ compatible with $p$. In other words, $f_t(z) = 0$ for any $z$ compatible with $p$. The base is proved.

Induction step is the same if $p$ ends in some other terminal. Now assume that $p$ ends in $v\in V(G)\setminus T(G)$. Take any $z\in f^{-1}(0)$ compatible with $p$. Let $p_j$ be the extension of $p$ by the edge which starts at $v$ and is labeled by $j\in\{0, 1, \ldots, k\}$. Next, let $v_j$ be the endpoint of $p_j$ (nodes $v_0, v_1, \ldots, v_k$ are successors of $v$). Since $S$ is winning, the number of $j\in\{0, 1, \ldots, k\}$ such that $z\in \mathcal{H}_j(p)$ is at least $k$. Hence by definition the number of $j\in \{0, 1, \ldots, k\}$ such that $z$ is compatible with $p_j$ is at least $k$. Finally, by the induction hypothesis this means that the number of $j\in\{0, 1, \ldots, k\}$ such that $f_{v_j}(z) = 0$ is at least $k$. On the other hand:
$$f_v = \THR^{k + 1}_2(f_{v_0}, f_{v_1}, \ldots, f_{v_k}).$$
Therefore $f_v(z) = 0$, as required.
\end{proof}

 One can formally define analogues notions for $R_k$-hypotheses games. We skip this as modifications are straightforwards and only formulate an analog of Proposition \ref{Q_k_equivalence}.

\begin{proposition}
\label{R_k_equivalence}
For all $f\in R_k, f\colon\{0,1\}^n\to\{0,1\}$ the following holds:
\begin{enumerate}[label=\textbf{(\alph*)}]
\item if $S$ if a Learner's winning strategy in $R_k$-hypotheses game for $f$, then its light form, considered as a circuit $C$ consisting only of $\THR^{k + 1}_2$ gates and literals, satisfies $C\le f$.

\item  Assume that $C\le f$ is a circuit, consisting only of $\THR^{k + 1}_2$ gates and literals. Then there exists a Learner's winning strategy $S$ in $R_k$-hypotheses game for $f$ such that the light form of $S$ coincides with $C$.
\end{enumerate}
\end{proposition}
\begin{remark}
It might be unclear why we prefer to construct strategies instead of constructing circuits directly, because beside the circuit itself we should also specify Learner's hypotheses. The reason is that strategies can be seen as \emph{proofs} that the circuit we construct is correct.
\end{remark}

\section{Results for Majority} \label{sec:majority}

\begin{proof}[Proof of Theorem \ref{cohen_conj_1}]
There exists an algorithm which in $n^{O(1)}$-time produces a monotone formula $F$ of depth $d = O(\log n)$ computing $\MAJ_{2n + 1}$. Below we will define a strategy $S_F$  in the $Q_2$-hypotheses game for $\MAJ_{2n + 1}$. Strategy $S_F$ will be winning for Learner. Moreover, its depth will be $d + O(\log n)$. In the end of the proof we will refer to Proposition \ref{Q_k_equivalence} to show that $S_F$ yields a $O(\log n)$-depth polynomial-time computable formula for $\MAJ_{2n + 1}$, consisting only of $\MAJ_3$ gates and variables.

 Strategy $S_F$ has two phases. The first phase does not uses $F$ at all, only the second phase does. The objective of the first phase is to find some distinct $i, j\in [2n + 1]$ such that either $z_i = 0 \land z_j = 1$ or $z_i = 1 \land z_j = 0$, where $z$ is the Nature's vector. This can be done as  follows.
\begin{lemma} 
\label{tree_lemma}
One can compute in polynomial time a 3-ary tree $T$ of depth $O(\log n)$ with the set of nodes $v(T)$ and a mapping $w\colon v(T) \to 2^{[2n + 1]}$ such that the following holds:
\begin{itemize}
\item if $r$ is the root of $T$, then $w(r) = [2n + 1]$;
\item if $v$ is not a leaf of $T$ and $v_1, v_2, v_3$ are $3$ children of $v$, then every element of $w(v)$ is covered at least twice by $w(v_1), w(v_2), w(v_3)$;
\item if $l$ is a leaf of $T$, then $w(r)$ is of size $2$.
\end{itemize}
\end{lemma}
\begin{proof}
We start with a trivial tree, consisting only of the root, to which we assign $[2n + 1]$. Then at each iteration we do the following. We have a $3$-ary tree in which nodes are assigned to some subsets of $[2n + 1]$. If every leaf is assigned to a set of size $2$, we terminate. Otherwise we pick any leaf $l$ of the current tree which is assigned to a subset $A\subset [2n + 1]$ of size at least $3$. We split $A$ into $3$ disjoint subsets $A_1, A_2, A_3$ of sizes $\lfloor |A|/3\rfloor, \lfloor |A|/3 \rfloor$ and $|A| - 2\lfloor |A|/3\rfloor$. We add $3$ children to $l$ (which become new leafs) and assign $A_1\cup A_2, A_1 \cup A_3, A_2 \cup A_3$ to them.

It is easy to verify that the sizes of $A_1\cup A_2, A_1 \cup A_3, A_2 \cup A_3$  are at least $2$ and at most $\frac{4}{5} \cdot |A|$. Hence the size of the set assigned to a node of depth $h$ is at most $\left(\frac{4}{5}\right)^h \cdot (2n + 1)$.  This means that the depth of the tree is at any moment at most $\log_{5/4}(2n + 1) = O(\log n)$. Therefore we terminate in $3^{O(\log n)} = n^{O(1)}$ iterations, as at each iterations we add $3$ new nodes. Each iteration obviously takes polynomial time.  
\end{proof}

We use $T$ to find two $i, j\in [2n + 1]$ such that either $z_i = 0$ or $z_j = 0$. Namely, we descend from the root of $T$ to one of its leafs. Learner maintains an invariant that the leftmost $0$-coordinate of $z$ is in $w(v)$, where $v$ is the current node of $T$. Let $v_1, v_2, v_3$ be $3$ children of $v$. Learner  for every $i\in [3]$ makes a hypothesis that the leftmost $0$-coordinate of $z$ is in $w(v_i)$. Due to the properties of $w$ at least two hypotheses are true. Nature indicates some $v_i$ for which this is true, and Learner descends to $v_i$. When Learner reaches a leaf, he knows a set of size two containing the leftmost $0$-coordinate of $z$. Let this set be $\{i,j\}$.

 We know that either $z_i$ or $z_j$ is $0$. Thus $z_i z_j \in\{00, 01, 10\}$. At the cost of one round we can ask Nature to identify an element of $\{00, 01, 10\}$ which differs from $z_i z_j$. If  $10$ is identified, then $z_i z_j \in \{00, 01\}$, and hence $z_i = 0$, i.e., we can already output $i$. Similar thing happens when $01$ is identified. Finally, if $00$ is identified, then the objective of the first phase is fulfilled and we can proceed to the second phase.

The second phase takes at most $d$ rounds. In this phase Learner produces a sequence $g_0, g_1, \ldots, g_{d^\prime}$, $d^\prime \le d$ of gates of $F$, where the depth of $g_i$ is $i$, the last gate $g_{d^\prime}$ is an input variable (i.e., a leaf of $F$) and each $g\in\{g_0, g_1, \ldots, g_{d^\prime}\}$ satisfies:
\begin{equation}
\label{maj_invariant}
\left(g(z) = 0 \land z_i z_j = 01 \right) \lor \left( g(\lnot z)  = 1 \land z_i z_j = 10\right).
\end{equation}
Here $\lnot z$ denotes the bit-wise negation of $z$. 

At the beginning Learner sets $g_0 = g_{\mathrm{out}}$ to be the output gate of $F$. Let us explain why \eqref{maj_invariant} holds for $g_{\mathrm{out}}$. Nature's vector is an element of $\MAJ_{2n + 1}^{-1}(0)$. I.e., the number of ones in $z$ is at most $n$. In turn, in $\lnot z$ there are at least  $n + 1$ ones. Since  $g_{\mathrm{out}}$ computes $\MAJ_{2n + 1}$, we have that  $g_{\mathrm{out}}(z) = 0$ and $g_{\mathrm{out}}(\lnot z) = 1$. In turn, by the first phase it is guarantied that $z_i z_j = 01 \lor z_i z_j = 10$.

Assume now that the second phase is finished, i.e., Learner has produced some $g_{d^\prime} = x_k$ satisfying \eqref{maj_invariant}. Then by \eqref{maj_invariant} either $g_{d^\prime}(z) = z_k = 0$ or $g_{d^\prime}(\lnot z) = (\lnot z)_k = 1$. In both cases $z_k = 0$, i.e., Learner can output $k$.

It remains to explain how to fulfill the second phase. It is enough to show the following. Assume that Learner knows a gate $g_l$ of $F$ of depth $l$  satisfying \eqref{maj_invariant} and that $g_l$ is not an input variable. Then in one round he can  either find a gate $g_{l + 1}$ of depth $l + 1$ satisfying \eqref{maj_invariant} or give a correct answer to the game.

The gate $g_{l + 1}$ will be one of the two gates which are fed to $g_l$.  Assume first that $g_l$ is an $\land$-gate and $g_l = u \land v$.
  From \eqref{maj_invariant} we conclude that from the following 3 statements exactly 1 is true for $z$:
\begin{align}
\label{1_pos}
u(z) &= 0 \mbox{ and } z_i z_j = 01, \\
\label{2_pos}
u(z) &= 1, v(z) = 0 \mbox{ and } z_i z_j = 01,\\
\label{3_pos}
u(\lnot z) &= v(\lnot z) = 1 \mbox{ and } z_i z_j = 10.
\end{align}
At the cost of one round Learner can ask Nature to indicate one statement which is false for $x$. If Nature says that \eqref{1_pos} is false for $z$, then   \eqref{maj_invariant} holds for $g_{l + 1} = v$. Next, if Nature says that \eqref{2_pos} is false for $z$, then  \eqref{maj_invariant} holds for $g_{l + 1} = u$. Finally, if Nature says that \eqref{3_pos} is false for $z$, then we know that $z_i z_j = 01$, i.e., Learner  can already output $i$.

In the same way we can deal with the case when $g_l$ is an $\lor$-gate and $g_l = u \lor v$. By \eqref{maj_invariant} exactly 1 of the following 3 statements is true for $z$:
\begin{align}
\label{4_pos}
u(z) &=   v(z) = 0 \mbox{ and } z_i z_j = 01, \\
\label{5_pos}
u(\lnot z) &= 1 \mbox{ and } z_i z_j = 10,\\
\label{6_pos}
u(\lnot z) &= 0,  v(\lnot z) = 1 \mbox{ and } z_i z_j = 10.
\end{align}
Similarly, Learner asks Nature to indicate one statement which is false for $z$. If Nature says that \eqref{4_pos} is false for $z$, then $z_i z_j = 10$, i.e., Learner can output $j$. Next, if Nature says that \eqref{5_pos} is false for $z$, then   \eqref{maj_invariant} holds for $g_{l + 1} = v$. Finally, if Nature says that \eqref{6_pos} is false for $z$, then \eqref{maj_invariant} holds for $g_{l + 1} = u$.

Thus $S_F$ is a $O(\log n)$-depth winning strategy of Learner. Apply Proposition \ref{Q_k_equivalence} to $S_F$. We get a $O(\log n)$-depth formula $F^\prime \le \MAJ_{2n + 1}$, consisting only of $\MAJ_3$ gates and variables. From the self-duality of $\MAJ_{2n + 1}$ and $\MAJ_3$ it follows that $F^\prime$ computes $\MAJ_{2n + 1}$. Finally, let us explain how to  compute $F^\prime$ in polynomial time. To do so we have to compute in polynomial time the light form of $S_F$, i.e., the underlying tree of $S_F$ and  the outputs of Learner in the leafs. It is easy to see that one can do this as follows.

First, compute $F$ and compute $T$ from Lemma \ref{tree_lemma}. For each leaf $l$ of $T$ do the following. Let $w(l) = \{i, j\}$. Add $3$ children to $l$. Two of them will be leafs of $S_F$, in  one Learner outputs $i$ and in the other Learner outputs  $j$. Attach a tree of $F$ to the third child.   Then  add to each non-leaf node of $F$ one more child so that now the tree of $F$ is $3$-ary. Each added child is a leaf of $S_F$. If a child was added to an $\land$-gate, then Learner outputs $i$ in this child. In turn, if a child was added to an $\lor$ gate, then Learner outputs $j$ in it. Finally, there are leafs that were in $F$ initially, each labeled by some input variable. In these nodes Learner outputs  the index of the corresponding input variable.

\end{proof}

\begin{proof}[Proof of Theorem \ref{transformation}]
How many rounds takes the first phase of the strategy $S_F$ from the previous proof? Initially the left-most $0$-coordinate of $z$ takes $O(n)$ values. At the cost of one round we can shrink the number of possible values almost by a factor of $3/2$. Thus the first phase corresponds to a ternary tree of depth $\log_{3/2}(n) + O(1)$. The size of that tree is hence $3^{\log_{3/2}(n) + O(1)} = O(n^{1/(1 - \log_3(2))}) = O(n^{2.70951\ldots})$. To some of its leafs we attach a tree of the same size as the initial formula $F$. As a result we obtain a formula $F^\prime$ of size $O(n^{2.70951\ldots} \cdot s)$ for $\MAJ_{2n + 1}$, consisting of $\MAJ_3$ gates and variables (here $s$ is the size of the initial formula $F$).

Let us show that we can perform the first phase in $\log_2(n) + O(1)$ rounds. This will improve the size of the previous construction to $O(3^{\log_2(n) + O(1)} \cdot s) = O(n^{\log_2(3)} \cdot s)$. However, the construction with $\log_2(n) + O(1)$ rounds will not be explicit. We need the following Lemma:

\begin{lemma}
\label{valiant_lemma}
There exists a formula $D$ with the following properties:
\begin{itemize}
\item formula $D$ is a complete ternary tree of depth $\lceil\log_2(n)\rceil + 10$;
\item every non-leaf node of $D$ contains a $\MAJ_3$ gate and every leaf of $D$ contains a conjunction of 2 variables;
\item $D(x) = 0$ for every $x\in \{0, 1\}^{2n + 1}$ with at most $n$ ones.
\end{itemize}
\end{lemma}
Let us at first explain how to use formula $D$ from Lemma \ref{valiant_lemma} to fulfill the first phase. Recall that our goal is to find two indices $i, j\in[2n + 1]$ such that either $z_i = 0$ or $z_j = 0$. To do so Learner descends from the output gate of $D$ to some of its leafs. He maintains an invariant that for his current gate $g$ of $D$ it holds that $g(z) = 0$. For the output gate the invariant is true because by Lemma~\ref{valiant_lemma} $D$ is $0$ on all Nature's possible vectors. If we reached a leaf so that $g$ is a conjuction of two variables $z_i$ and $z_j$, then the first phase is fulfilled (by the invariant $z_i\land z_j = 0$). Finally, if $g$ is a non-leaf node of $D$, i.e., a $\MAJ_3$ gate, then we can descend to one of the children of $g$ at the cost of one round without violating the invariant. Indeed, as $g(z) = 0$, then the same is true for at least $2$ children of $g$. For each child $g_i$ of $g$ Learner makes a hypotheses that $g_i(z) = 0$. Any Nature's response allows us to replace $g$ by some $g_i$. 

\begin{proof}[Proof of Lemma \ref{valiant_lemma}]
We will show existence of such $D$ via probabilistic method. Namely, independently for each leaf $l$ of $D$ choose $(i, j) \in [2n + 1]^2$ uniformly at random and put the  conjuction $z_i \land z_j$ into $l$. It is enough to demonstrate that for any $x\in\{0, 1\}^{2n + 1}$ with at most $n$ ones it hols that $\Pr[D(x) = 1] < 2^{-2n - 1}$. 

To do so we use the modification of the standard Valiant's argument. For any fixed $x$ let $p$ be the probability that a leaf $l$ of $D$ equals $1$ on $x$. This probability is the same for all the leafs and is at most $1/4$. Now, $\Pr[D(x) = 1]$ can be expressed exactly in terms of $p$ as follows:
$$\Pr[D(x) = 1] = \underbrace{f(f(f(\ldots f}_{\text{$\lceil\log_2(n)\rceil + 10$}}(p)))\ldots),$$
where $f(t) = t^3 + 3 t^2 (1 - t) = 3t^2 -2t^3$. Observe that $3f(t) \le (3t)^2$. Hence
$$3\Pr[D(x) = 1] \le (3p)^{2^{\lceil\log_2(n)\rceil + 10}}\le (3/4)^{1000 n} < (1/2)^{-2n - 1}.$$
\end{proof}
\end{proof}

\section{Proof of the main theorem} \label{sec:proof_main}

Theorem \ref{Q_k_main} follows from Proposition \ref{Q_k_to_protocols} (Subsection~\ref{subsec:circ_to_prot}) and Proposition \ref{Q_k_to_circuits} (Subsection~\ref{subsec:prot_to_circ}).  In turn, Theorem \ref{R_k_main} follows from Proposition \ref{R_k_to_protocols} (Subsection~\ref{subsec:circ_to_prot}) and Proposition \ref{R_k_to_circuits} (Subsection~\ref{subsec:prot_to_circ}).

\subsection{From circuits to protocols} \label{subsec:circ_to_prot}
\begin{proposition}
\label{Q_k_to_protocols}
For any constant $k\ge 2$ the following holds.
Assume that $f\in Q_k$ and $C\le f$ is a circuit, consisting only of $\THR^{k + 1}_2$ gates and variables. Then there is a protocol $\pi$, computing $Q_k$-communication game for $f$, such that $\mathrm{depth}(\pi) =  O(\mathrm{depth}(C))$.
\end{proposition}
\begin{proof}
Let the inputs to parties be $z^1, \ldots, z^k\in f^{-1}(0)$.
Parties descend from the output gate of $C$ to one of the inputs. They maintain the invariant that for the current gate $g$ of $C$ it holds that $g(z^1) = g(z^2) = \ldots = g(z^k) = 0$. If $g$ is not yet an input, then $g$ is a $\THR^{k + 1}_2$ gate and $g = \THR^{k + 1}_2(g_1, \ldots, g_{k + 1})$ for some gates $g_1, \ldots, g_{k+1}$. For each $z^i$ we have $g(z^i) =  \THR^{k + 1}_2(g_1(z^i), \ldots, g_{k + 1}(z^i)) = 0$. Hence for each $z^i$ there is at most one gate out of $g_1, \ldots, g_{k+1}$ satisfying $g_j(z^i) = 0$. Hence in $O(1)$ bits of communication parties can agree on the index $j\in [k + 1]$ satisfying $g_j(z^1) = g_j(z^2) = \ldots g_j(z^k) = 0$. 

Thus in $O(\mathrm{depth}(\pi))$ bits of communication they reach some input of $C$. If this input contains the variable $x_l$, then by the invariant $z^1_l = z^2_l = \ldots = z_l^k = 0$, as required.
\end{proof}

Exactly the same argument can be applied to the following proposition.
\begin{proposition}
\label{R_k_to_protocols}
For any constant $k\ge 2$ the following holds.
Assume that $f\in R_k$ and $C\le f$ is a circuit, consisting only of $\THR^{k + 1}_2$ gates and literals. Then there is a protocol $\pi$, computing $R_k$-communication game for $f$, such that $\mathrm{depth}(\pi) =  O(\mathrm{depth}(C))$.
\end{proposition}

\subsection{From protocols to circuits} \label{subsec:prot_to_circ}
\begin{proposition}
\label{Q_k_to_circuits}
For every constant $k\ge 2$ the following holds. Let $f\in Q_k$. Assume that $\pi$ is a communication protocol computing $Q_k$-communication game for $f$. Then there is a circuit $C \le f$, consisting of $\THR^{k + 1}_2$ gates and variables, such that $\mathrm{depth}(C) = O(\mathrm{depth}(\pi))$.
\end{proposition}
\begin{proof}
In the proof we will use the following terminology for strategies in $Q_k$-hypotheses game. Fix some strategy $S$. \emph{A current play} is a finite sequence $r_1, r_2, r_3, \ldots r_j$ of integers from $0$ to $k$. By $r_i$ we mean Nature's response in the $i$th round. Given a current play, let $\mathcal{H}^i_0, \ldots, \mathcal{H}^i_{k}\subset f^{-1}(0)$ be $k + 1$ hypotheses Learner makes in the $i$th round according to $S$ if Nature's responses in the first $i - 1$ rounds were $r_1, \ldots, r_{i - 1}$. If after that Nature's response is $r_i$, then Nature's input vector $z$ satisfies $z\in H^i_{r_i}$. We say that $z\in f^{-1}(0)$ is \emph{compatible} with the current play $r_1, \ldots, r_j$  if $z\in H^1_{r_1}, \ldots, z\in H^j_{r_j}$. Informally, this means that Nature, having $z$ on input, can  produce responses $r_1, \ldots, r_j$ by playing against strategy $S$.

Set $d = \mathrm{depth}(\pi)$. By Proposition \ref{Q_k_equivalence}
it is enough to give a $O(d)$-round winning strategy of Learner in the $Q_k$-hypotheses game for $f$. Strategy proceeds in $d$ iterations, each iteration takes $O(1)$ rounds.

As the game goes on, a sequence of Nature's responses $r_1, r_2, r_3\ldots$ is produced. Assume that $r_1, \ldots, r_{h^\prime}$ are Nature's responses in the first $h$ iteration (here $h'$ is the number of rounds in the first $h$ iterations).
Given any $r_1, r_2, r_3\ldots$, by $\mathcal{Z}_h$ we denote the set of all $z\in f^{-1}(0)$ which are compatible with $r_1, \ldots r_{h^\prime}$, . We also say that elements of $\mathcal{Z}_h$ are \emph{compatible with the current play after $h$ iterations}.

Let $V$ be the set of all nodes of the protocol $\pi$ and let $T$ be the set of all terminals of the protocol $\pi$.

Consider a set $\mathcal{Z} \subset f^{-1}(0)$, a set of nodes $U\subset V$ and a function $g\colon \mathcal{Z} \to C$, where $|C| = k$.
 \emph{A $g$-profile} of a tuple $(z^1, \ldots, z^k) \in \mathcal{Z}$  is a vector $(g(z^1), \ldots, g(z^k)) \in C^k$. 

We say that $g\colon \mathcal{Z} \to C$ is \emph{complete} for $\mathcal{Z}$ with respect to the set of nodes $U$ if the following holds. For every vector $\bar{c}\in C^k$ there exists a node $v\in U$ such that all tuples from $\mathcal{Z}^k$ with $g$-profile $\bar{c}$ visit $v$ in the protocol $\pi$.

We say that a set of nodes $U\subset T$ is complete for $\mathcal{Z}$ if there exists $g\colon \mathcal{Z} \to C$, $|C| = k$ which is complete for $\mathcal{Z}$ with respect to $U$.

Note that we can consider only complete sets of size at most $k^k$. Formally, if $U$ is complete for $\mathcal{Z}$, then there is a subset $U^\prime\subset U$ of size at most $k^k$ which is also complete for $\mathcal{Z}$. Indeed, there are $k^k$ possible $g$-profiles and for each we need only one node in $U$.

\begin{lemma}
\label{complete_lemma}
Assume that $U\subset T$ is complete for $\mathcal{Z} \subset f^{-1}(0)$. Then there exists $i\in [n]$ such that $z_i = 0$ for every $z\in\mathcal{Z}$.
\end{lemma}
\begin{proof}

If $\mathcal{Z}$ is empty, then there is nothing to prove. Otherwise let $g\colon\mathcal{Z} \to C$, $|C| = k$ be complete for $\mathcal{Z}$ with respect to $U$. Take any vector $\bar{c} = (c_1, \ldots, c_k) \in C^k$  such that $\{ c_i \mid i\in [k]\} = g(\mathcal{Z})$. There exists a node $v\in U$ such that any tuple from $\mathcal{Z}^k$  with $g$-profile $\bar{c}$ visits $v$. Note that $v$ is a terminal of $\pi$ and let $i$ be the output of $\pi$ in $g$. Let us show that for any $z\in\mathcal{Z}$ it holds that $z_i = 0$. Indeed, note that there exists a tuple  $\bar{z}\in\mathcal{Z}^k$ which includes $z$  and which has $g$-profile $\bar{c}$.  This tuple visits $v$. Since $\pi$ computes $Q_k$-communication game for $f$, every element of the tuple $\bar{z}$ should have $0$ at the $i$th coordinate. In particular, this holds for $z$.
\end{proof}
After $d$ iterations Learner should be able to produce an output. For that there should exist $i\in [n]$ such that for any $z\in\mathcal{Z}_d$ it holds that $z_i = 0$. We will use Lemma \ref{complete_lemma} to ensure that. Namely, we will ensure that there exists $U\subset T$ which is complete for $\mathcal{Z}_d$. Learner achieves this by maintaining the following invariant. 

Let us say that a set of nodes $U$ is $h$-low if every element of $U$ is either a terminal or a node of depth at least $h$.

\begin{algorithm}
\caption{}
\label{euclid}
There is a $h$-low set $U$ which is complete for $\mathcal{Z}_h$.
\end{algorithm}
This invariant implies that Learner wins in the end, as any $d$-low set consists only of terminals.

A $0$-low set which is complete for $\mathcal{Z}_0 = f^{-1}(0)$ is a set consisting only of the starting node of $\pi$.

Assume that \textbf{Invariant 1} holds after $h$ iteration. Let us show how to perform the next iteration to maintain the invariant. For that we need a notion of \emph{communication profile}.

A communication profile of $z\in f^{-1}(0)$ with respect to a set of nodes $U\subset V$ is a function $p_z\colon U \to \{0, 1\}$. For $v\in U$ the value of $p_z(v)$ is defined as follows. If $v$ is a terminal, set $p_z(v) = 0$. Otherwise let $i\in [k]$ be the index of the party communicating at $v$. Set $p_z(v)$ to be the bit transmitted by the $i$th party at $v$ on input $z$. I.e., $p_z$ for every $v\in U$ contains  information where the protocol goes from the node $v$ if the party, communicating at $v$, has $z$ on input. 

We also define a communication profile of the tuple $(z^1, \ldots, z^k) \in (f^{-1}(0))^k$ as $(p_{z^1}, \ldots, p_{z^k})$.

\begin{lemma}
\label{communication_profile_lemma}
Let $(z^1, \ldots, z^k), (y^1, \ldots, y^k)  \in (f^{-1}(0))^k$ be two inputs visiting the same node $v\in V\setminus T$. Assume that their communication profiles with respect to $\{v\}$ coincide. Then these two inputs visit the same successor of $v$. 
\end{lemma}
\begin{proof}
Let their common communication profile with respect to $\{v\}$ be $(p_1, \ldots, p_k)$. Next, assume that $i$ is the index of the party communicating at $v$. Then the information where these inputs descend from $v$ is contained in $p_i$.
\end{proof}

Here is what Learner does during the $(h + 1)$st iteration. He takes any $h$-low $U$ of size at most $k^k$ which is complete for $\mathcal{Z}_h$. Then he takes any $g\colon \mathcal{Z}_h\to C$, $|C| = k$ which is complete for $\mathcal{Z}_h$ with respect to $U$. He now devises a new function $g^\prime$ taking elements of  the set $\mathcal{Z}_h$ on input. The value of $g^\prime(z)$ is a pair $(p_z, g(z))$, where $p_z$ is a communication profile of $z$ with respect to $U$. There are at most $2^{|U|} \le 2^{k^k}$ different communication profiles with respect to $U$. Hence $g^\prime(z)$ takes at most $2^{k^k} \cdot k = O(1)$ values.

At each round of the $(h + 1)$st iteration Learner asks Nature to identify some pair $(p, c)$, where $p\colon U \to \{0, 1\}$ and $c\in C$, such that $g^\prime(z) \neq (p, c)$
for the Nature's vector $z$. Namely, we take any $k + 1$ values of $g^\prime$ which are not yet rejected by Nature and ask Nature to reject one of them. We do so  until there are only $k$ possible values $(p_1, c_1), \ldots (p_k, c_k)$ left. This takes $O(1)$ rounds and the $(h + 1)$st iteration is finished. Any $z\in f^{-1}(0)$ which is compatible with the responses  Nature' gave during the $(h+1)$st iteration in the current play satisfies $g^\prime(z) \in C^\prime = \{(p_1, c_1), \ldots (p_k, c_k)\}$. In particular, any $z\in\mathcal{Z}_{h + 1}$ satisfies $g^\prime(z) \in C^\prime$. I.e., the restriction of $g^\prime$ to $\mathcal{Z}_{h + 1}$ is a function of the form $g^\prime\colon \mathcal{Z}_{h + 1} \to C^\prime$.
 Let us show that $g^\prime\colon \mathcal{Z}_{h + 1} \to C^\prime$  is complete for $\mathcal{Z}_{h + 1}$ with respect to some $(h + 1)$-low set $U^\prime$. This will ensure that \textbf{Invariant 1} is maintained after $h + 1$ iterations.

We define $U^\prime$ is follows. Take any vector $\bar{c} \in (C^\prime)^k$. It is enough to show that all the inputs from $(\mathcal{Z}_{h + 1})^k$ with  $g^\prime$-profile $\bar{c}$ visit the same node $v^\prime$ which is either a terminal or of depth at least $h + 1$. Then we just set $U^\prime$ to be the union of all such $v^\prime$ over all possible $g^\prime$-profiles.

  All the tuples from $(\mathcal{Z}_{h + 1})^k$ with the same $g^\prime$-profile visit the same node $v\in U$. This is because $g^\prime$-profile of a tuple determines its $g$-profile (the value of $g^\prime$ determines the value of $g$) , and hence we can use \textbf{Invariant 1} for $\mathcal{Z}_{h - 1}$ here.  If $v$ is a terminal, there is nothing left to prove. Otherwise, note that $g^\prime$-profile of a tuple also determines its communication profile with respect to $U$ and hence with respect to $\{v\} \subset U$. Therefore all the tuples with the same $g^\prime$-profile by Lemma \ref{communication_profile_lemma} visit the same successor of $v$. 
\end{proof}

With straightforward modifications one can obtain a proof of the following:
\begin{proposition}
\label{R_k_to_circuits}
For every constant $k\ge 2$ the following holds. Let $f\in R_k$. Assume that $\pi$ is a dag-like protocol computing $R_k$-communication game for $f$. Then there is a circuit $C \le f$, consisting of $\THR^{k + 1}_2$ gates and literals, satisfying $\mathrm{depth}(C) = O(\mathrm{depth}(\pi))$.
\end{proposition}

\begin{corollary}[Weak version of Theorem \ref{cohen_2}]
\label{weak}
For any constant $k \ge 2$ there exists $O(\log^2 n)$-depth formula for $\THR^{kn + 1}_{n + 1}$, consisting only of $\THR^{k + 1}_2$ gates and variables.
\end{corollary}
\begin{proof}
We will show that there exists $O(\log^2 n)$-depth protocol $\pi$ computing $Q_k$-communication game for $\THR^{kn + 1}_{n + 1}$. By Proposition \ref{Q_k_to_circuits} this means that there is a  $O(\log^2 n)$-depth formula $F\le \THR^{kn + 1}_{n + 1}$, consisting only of $\THR^{k + 1}_2$ gates and variables. It is easy to see that $F$ actually coincides with $\THR^{kn + 1}_{n + 1}$. Indeed, assume that $F(x) = 0$ for some $x$ with at least $n + 1$ ones. Then it is easy to construct $x^2, \ldots, x^k$, each with $n$ ones, such that there is no common $0$-coordinate for $x, x_2, \ldots, x_k$. On all of these vectors $F$ takes value $0$. However, the function computed by $F$ should belong to $Q_k$ (Proposition \ref{Q_k_char}).

Let $\pi$ be the following protocol. Assume that the inputs to parties are $x^1, x^2, \ldots, x^k \in \{0, 1\}^{kn + 1}$, without loss of generality we can assume that in each $x^r$ there are exactly $n$ ones. For $x\in\{0, 1\}^{kn + 1}$ define $\mathrm{supp}(x) = \{i \in [kn + 1] \mid x_i = 1\}$. Let $T$ be a binary rooted tree of depth $d = \log_2(n) + O(1)$ with $kn + 1$ leafs. Identify leafs of $T$ with elements of $[kn + 1]$. For a node $v$ of $T$ let $T_v$ be the set of all leafs of $T$ which are descendants of $v$. Once again, we view $T_v$ as a subset of $[kn + 1]$.

The protocol proceeds in at most $d$ iterations. After $i$ iterations, $i = 0, 1, 2, \ldots, d$, parties agree on a node $v$ of $T$ of depth $i$, satisfying the following invariant:
\begin{equation}
\label{binary_search_invariant}
\sum\limits_{r = 1}^k \left|\mathrm{supp}(x^r) \cap T_v\right| < |T_v|. 
\end{equation}
At the beginning Invariant \eqref{binary_search_invariant} holds just because $v$ is the root, $T_v = [kn + 1]$ and each $\mathrm{supp}(x^r)$ is of size  $n$. 

After $d$ iterations $v = l$ is a leaf of $T$. Parties output $l$. This is correct because by \eqref{binary_search_invariant} we have $|T_l| = 1 \implies  |\mathrm{supp}(x^r) \cap T_l| = 0 \implies x_l = 0$ for every $r\in [k]$. 

Let us now explain what parties do at each iteration. If the current $v$ is not a leaf,  let $v_0, v_1$ be two children of $v$. Each party sends $\left|\mathrm{supp}(x^r) \cap T_{v_0}\right|$ and $\left|\mathrm{supp}(x^r) \cap T_{v_1}\right|$, using $O(\log n)$ bits. Since $T_{v_0}$ and $T_{v_1}$ is a partition of $T_v$, we have:
$$
 \sum\limits_{b = 0}^1 \sum\limits_{r = 1}^k \left|\mathrm{supp}(x^r) \cap T_{v_b}\right| = \sum\limits_{r = 1}^k \left|\mathrm{supp}(x^r) \cap T_v\right| < |T_v| = \sum\limits_{b = 0}^1 |T_{v_b}|.
$$
Thus the inequality:
\begin{equation}
\label{b_ineq}
\sum\limits_{r = 1}^k \left|\mathrm{supp}(x^r) \cap T_{v_b}\right| < |T_{v_b}|
\end{equation}
is true either for $b = 0$ or for $b = 1$. Let $b^*$ be the smallest $b\in\{0, 1\}$ for which \eqref{b_ineq} is true. Parties proceed to the next iteration with $v$ being replaced by $v_{b^*}$.

There are $d = O(\log n)$ iterations, at each parties communicate $O(\log n)$ bits. Hence $\pi$ is $O(\log^2 n)$-depth, as required.
\end{proof}


\begin{remark}
Strategy from the proof of Proposition \ref{Q_k_to_circuits} is efficient only in terms of the number of rounds. 
In the next section we give another version of this strategy. This  version will ensure that circuits we obtain  from protocols for $Q_k$-communication games  are not only low-depth, but also polynomial-size and explicit.  For that, however, we require a bit more from the protocol $\pi$. 
\end{remark}

\section{Effective version} \label{sec:effective}

Fix $f\in Q_k$. We say that a dag-like communication protocol $\pi$ \emph{strongly} computes $Q_k$-communication game for $f$ if for every terminal $t$ of $\pi$, for every $x\in f^{-1}(0)$ and for every $i\in [k]$ the following holds. If $x$ is $i$-compatible with $t$, then $x_j = 0$, where $j = l(t)$ is the label of terminal $t$ in the protocol $\pi$.

Similarly, fix $f\in R_k$. We say that a dag-like communication protocol $\pi$ \emph{strongly} computes $R_k$-communication game for $f$ if for every terminal $t$ of $\pi$, for every $x\in f^{-1}(0)$ and for every $i\in [k]$ the following holds. If $x$ is $i$-compatible with $t$, then $x_j = b$, where $(j, b) = l(t)$ is the label of terminal $t$ in the protocol $\pi$.

Strong computability essentially (but not completely) coincides with the notion of computability that Sokolov gave in~\cite{Sok2017} for general relations. Strong computability implies more intuitive notion of computability that we gave in the Preliminaries. The opposite direction is false in general. 

Next we prove an effective version of Proposition~\ref{Q_k_to_circuits}.
\begin{theorem}
\label{efficient_Q_k}
For every constant $k\ge 2$ there exists  a polynomial-time algorithm $A$ such that the following holds. Assume that $f\in Q_k$ and $\pi$ is a dag-like protocol which strongly computes  $Q_k$-communication game for $f$. Then, given the light form of $\pi$, the algorithm $A$ outputs a circuit $C \le f$, consisting only of $\THR^{k + 1}_2$ gates and variables, such that $\mathrm{depth}(C) = O(\mathrm{depth}(\pi))$, $\mathrm{size}(C) = O\left(\mathrm{size}(\pi)^{O(1)}\right)$. 
\end{theorem}
\begin{proof}
We will again give a $O(d)$-round winning strategy of Learner in the $Q_k$-hypotheses game for $f$. Now, however, we should ensure that the light form of our strategy is of size $ O\left(\mathrm{size}(\pi)^{O(1)}\right)$ and can be computed in time $O\left(\mathrm{size}(\pi)^{O(1)}\right)$ from the light form of $\pi$. Instead of specifying the light form of our strategy directly we will use the following trick. Assume that Learner has a \emph{working tape} consisting of $O(\log \mathrm{size}(\pi))$ cells, where each cell can store one bit. Learner memorizes all the Nature's responses so that he knows the current position of the game. But he \emph{does not} store the sequence of Nature's responses on the working tape (there is no space for it). Instead, he first makes his hypotheses which depend on the current position. Then he receives a Nature's response $r\in \{0, 1, \ldots, k\}$. And then he \emph{modifies} the working tape, but the result should depend only on the current content of the working tape and on $r$ (and not on the current position in a game). Moreover, we will ensure that modifying the working tape takes  $O\left(\mathrm{size}(\pi)^{O(1)}\right)$ time, given the light form of $\pi$. 

The main purpose of the working tape manifests itself in the end. Namely, at some point Learner decides to stop making hypotheses. This should be indicated on the working tape. More importantly, Learner's output should depend only on the content of working tape in the end (and not on the whole sequence of Nature's responses). Moreover, this should take $O\left(\mathrm{size}(\pi)^{O(1)}\right)$ time to compute that output, given the light form of $\pi$. 

If a strategy satisfies these restrictions, then its light form is computable in  $O\left(\mathrm{size}(\pi)^{O(1)}\right)$ time given the light form of $\pi$. Indeed, the underlying dag will consist of all possible configurations of the working tape. There are $O\left(\mathrm{size}(\pi)^{O(1)}\right)$ of them, as working tape uses $O(\log \mathrm{size}(\pi))$ bits. For all non-terminal configurations $c$ we go through all $r\in\{0, 1, \ldots, k\}$. We compute what would be a configuration $c_r$ of the working tape if the current configuration is $c$ and Nature's response is $r$. After that we connect $c$ to $c_0, c_1, \ldots, c_k$. Finally, in all terminal configurations we compute the outputs of Learner. This gives a light form of our strategy in $O\left(\mathrm{size}(\pi)^{O(1)}\right)$ time.

Let $V$ be the set of nodes of $\pi$ and $T$ be the set of terminals of $\pi$.
Strategy proceeds in $d$ iterations, each taking $O(1)$ rounds. We define sets $\mathcal{Z}_h$ exactly as in the proof of Proposition \ref{Q_k_to_circuits}. We also use the same notion of communication profile. However, we define completeness in a different way. First of all, instead of working with sets of nodes with no additional structure we will work with \emph{multidimensional arrays} of nodes. Namely, we will consider $k$-dimensional arrays in which every dimension is indexed by integers from $[k]$. Formally, such arrays are functions of the form $M\colon [k]^k\to V$. We will use notation $M[c_1, \ldots, c_k]$ for the value of $M$ on $(c_1, \ldots, c_k) \in [k]^k$.

Consider any $\mathcal{Z}\subset f^{-1}(0)$.
We say that $g\colon \mathcal{Z} \to [k]$ is complete for $\mathcal{Z}$ with respect to a multidimensional array $M\colon [k]^k \to V$ if for every $(c_1, \ldots, c_k) \in [k]^k$, for every $i\in [k]$ and for every $z\in\mathcal{Z}$ the following holds. If $c_i = g(z)$, then $z$ is $i$-compatible with $M[c_1, \ldots, c_k]$.

We say that a multidimensional array $M\colon [k]^k \to V$ is complete for $\mathcal{Z}$ if there exists $g\colon\mathcal{Z} \to [k]$ which is complete with respect to $M$.

To digest the notion of completeness it is instructive to consider the case $k = 2$. In this case $M$ is a $2\times 2$ table containing four nodes of $\pi$. The function $g\colon\mathcal{Z}\to [2]$ is complete for $\mathcal{Z}$ with respect to $M$ if the following holds. First, for every $z\in\mathcal{Z}$  two nodes in the $g(z)$th \emph{row} of $M$ should be \emph{$1$-compatible} with $z$. Second,  for every $z\in\mathcal{Z}$  two nodes in the $g(z)$th \emph{column} of $M$ should be \emph{$2$-compatible} with $z$.

Let us now establish an analog of Lemma \ref{complete_lemma}.

\begin{lemma}
\label{complete_lemma_2}
Assume that $M\colon [k]^k \to T$ is complete for $\mathcal{Z}\subset f^{-1}(0)$. Let $l$ be the output of $\pi$ in the terminal $M[1, 2, \ldots, k]$. Then $z_l = 0$ for every $\mathcal{Z}$.
\end{lemma}
\begin{proof}
Since $\pi$ strongly computes $Q_k$-communication game for $f$, it is enough to show that every $z\in\mathcal{Z}$ is $i$-compatible with  $M[1, 2, \ldots, k]$ for some $i$. Take $g\colon \mathcal{Z} \to [k]$ which is complete for $\mathcal{Z}$ with respect to $M$. By definition $z$ is $g(z)$-compatible with $M[1, 2, \ldots, k]$.
\end{proof}

We now proceed to the description of the Learner's strategy. The working tape of Learner consists of:
\begin{itemize}
\item an integer $iter$;
\item a multidimensional array $M\colon [k]^k \to V$;
\item $O(1)$ additional bits of memory. 
\end{itemize}
Integer $iter$ will be at most $d\le \mathrm{size}(\pi)$ so to store all this information we need $O(\log (\mathrm{size}(\pi)))$ bits, as required. 
Integer $iter$ always equals the number of iterations performed so far (at the beginning $iter = 0$). The array $M$ changes  only at the moments when $iter$ is incremented by $1$. So let $M_h$ denote the content of the array $M$ when $iter = h$. 

We call an array of nodes $h$-low if every node in it is either terminal or of depth at least $h$. Learner maintains the following invariant.

\begin{algorithm}[h!]
\caption{}
\label{euclid}
$M_h$ is $h$-low and $M_h$ is complete for $\mathcal{Z}_h$.
\end{algorithm}

At the beginning Learner sets every element of  $M_0$ to be the starting node of $\pi$ so that \textbf{Invariant 2} trivially holds.

Note that every node in $M_d$ is a terminal of $\pi$. After $d$ iterations Learner outputs the label of terminal $M_d[1, 2, \ldots, k]$ in the protocol $\pi$. As $M_d$ is complete for $\mathcal{Z}_d$ due to \textbf{Invariant 2}, this by Lemma \ref{complete_lemma_2} will be a correct output in the $Q_k$-hypotheses game for $f$. Obviously producing the output takes polynomial time given the light form of $\pi$ and the content of Learner's working tape in the end. 

Now we need to perform an iteration. Assume that $h$ iterations passed and \textbf{Invariant 2} still holds. Let $U_h$ be the set of all nodes appearing in $M_h$. Take any function $g\colon \mathcal{Z}_h \to [k]$ which is complete for $\mathcal{Z}_h$ with respect to $M_h$.

At each round of the $(h + 1)$st iteration Learner asks Nature to specify some pair $(p, c) \in \{0, 1\}^{U_h} \times [k]$ such that $(p_z, g(z)) \neq (p, c)$, where $z$ is the Nature's vector and $p_z$ is a communication profile of $z$ with respect to $U_h$. Learner stores each $(p, c)$ using his $O(1)$ additional bits on the working tape. Learner can do this until there are only $k$ pairs from $(p_1, c_1), \ldots, (p_k, c_k)\in \{0, 1\}^{U_h} \times [k]$ left which are not rejected by Nature. When this moment is reached, the $(h + 1)$st iteration is finished. The iteration takes $2^{|U_h|} \cdot k - k = O(1)$ rounds, as required. For any $z$ compatible with the current play after $h + 1$ iterations we know that $(p_z, g(z))$ is among $(p_1, c_1), \ldots, (p_k, c_k)$, i.e,   
\begin{equation}
\label{h_equation}
(p_z, g(z)) \in \{(p_1, c_1), \ldots, (p_k, c_k)\} \mbox{ for all } z\in\mathcal{Z}_{h + 1}.
\end{equation}

Learner writes $(p_1, c_1), \ldots, (p_k, c_k)$ on the working tape (all the pairs that were excluded are on the working tape and hence he can compute the remaining ones). Learner then computes a $(h+1)$-low array $M_{h + 1}$ which will be complete for $\mathcal{Z}_{h + 1}$. To compute $M_{h + 1}$ he will only need to know $M_h$, $(p_1, c_1), \ldots, (p_k, c_k)$ (this information is on the working tape) and the light form of $\pi$. 

Namely, Learner determines $M_{h + 1}[d_1, \ldots, d_k]$ for $(d_1, \ldots, d_k) \in [k]^k$ as follows. Consider the node $v = M_h[c_{d_1}, \ldots, c_{d_k}]$. If $v$ is a terminal, then set $M_{h + 1}[d_1, \ldots, d_k] = v$. Otherwise let $i\in [k]$ be the index of the party communicating at $v$. Look at $p_{d_i}$, which can be considered as a function of the form $p_{d_i}\colon U_h \to \{0, 1\}$. Define $r = p_{d_i}(v)$. Among two edges, starting at $v$, choose one which is labeled by $r$. Descend along this edge from $v$ and let the resulting successor of $v$ be $M_{h + 1}[d_1, \ldots, d_k]$.

Obviously, computing $M_{h + 1}$ takes $O\left(\mathrm{size}(\pi)^{O(1)}\right)$.
To show that \textbf{Invariant 2} is maintained we have to show that \textbf{\emph{(a)}} $M_{h + 1}$ is $(h + 1)$-low and \textbf{\emph{(b)}} $M_{h + 1}$ is complete for $\mathcal{Z}_{h + 1}$.

The first part, \textbf{\emph{(a)}}, holds because each $M_{h + 1}[d_1, \ldots, d_k]$ is either a terminal or a successor of a node of depth at least $h$. For \textbf{\emph{(b)}}  we define the following function:
$$g^\prime\colon\mathcal{Z}_{h + 1}\to [k], \qquad g^\prime(z) = i, \mbox{ where $i$ is such that } (p_z, g(z)) = (p_i, c_i).$$
By \eqref{h_equation} this definition is correct. We will show that $g^\prime$ is complete for $\mathcal{Z}_{h + 1}$ with respect to $M_{h + 1}$.

For that take any $(d_1, \ldots, d_k) \in [k]^k, z\in\mathcal{Z}_{h + 1}$ and $i\in [k]$ such that $d_i = g^\prime(z)$. We shall show that $z$ is $i$-compatible with a node $M_{h + 1}[d_1, \ldots, d_k]$. By definition of $g^\prime$ we have that $g(z) = c_{d_i}$. As by \textbf{Invariant 2} function $g$   is complete for $\mathcal{Z}_h$ with respect to $M_h$, this means that $z$ is $i$-compatible with $v = M[c_{d_1}, \ldots, c_{d_k}]$. If $v$ is a terminal, then $M_{h+1}[d_1, \ldots, d_k] = v$ and there is nothing left to proof.

 Otherwise $v\in V\setminus T$. Let $j$ be the index of the party communicating at $v$. By definition $M_{h + 1}[d_1, \ldots, d_k]$ is a successor of $v$. If $j\neq i$, i.e., not the $i$th party communicates at $v$, then any successor of $v$ is $i$-compatible with $z$. Finally, assume that $j = i$. Node $M_{h + 1}[d_1, \ldots, d_k]$ is obtained from $v$ by descending along the edge which is labeled by $r = p_{d_i}(v)$. Hence to show that $z$ is $i$-compatible with $M_{h + 1}[d_1, \ldots, d_k]$ we should verify that at $v$ on input $z$ the $i$th party transmits the bit $r$. For that again recall that $g^\prime(z) = d_i$, which means by definition of $g^\prime$ that $p_z = p_{d_i}$. I.e., $p_{d_i}$ is the communication profile of $z$ with respect to $U_h$. In particular, the value $r = p_{d_i}(v)$ is the bit transmitted by the $i$th party on input $z$ at $v$, as required.
\end{proof}

In the same way one can obtain an analog of the previous theorem for the $R_k$-case.
\begin{theorem}
\label{efficient_R_k}
For every constant $k\ge 2$ there exists  a polynomial-time algorithm $A$ such that the following holds. Assume that $f\in R_k$ and $\pi$ is a dag-like protocol which strongly computes  $R_k$-communication game for $f$. Then, given the light form of $\pi$, the algorithm $A$ outputs a circuit $C \le f$, consisting only of $\THR^{k + 1}_2$ gates and literals, such that $\mathrm{depth}(C) = O(\mathrm{depth}(\pi))$, $\mathrm{size}(C) = O\left(\mathrm{size}(\pi)^{O(1)}\right)$. 
\end{theorem}

\section{Derivation of Theorems \ref{cohen_conj_1} and \ref{cohen_2}} \label{sec:corol}

In this section we obtain Theorems \ref{cohen_conj_1} and \ref{cohen_2} by devising protocols  strongly computing the corresponding $Q_k$-communication games. Unfortunately, establishing strong computability requires diving into straightforward but tedious technical details, even for simple protocols.  

\begin{proof}[Alternative proof of Theorem \ref{cohen_conj_1}]
 We will show that there exists $O(\log n)$-depth protocol $\pi$ with polynomial-time computable light form, strongly computing $Q_2$-communication game for $\MAJ_{2n + 1}$. By Theorem \ref{efficient_Q_k} this means that there is a polynomial-time computable $O(\log n)$-depth formula $F \le \MAJ_{2n + 1}$, consisting only of $\MAJ_3$ gates and variables. From self-duality of $\MAJ_{2n + 1}$ and $\MAJ_3$ it follows that $F$ computes $\MAJ_{2n + 1}$. 

Take a polynomial-time computable $O(\log n)$-depth monotone formula $F^\prime$ for $\MAJ_{2n  + 1}$. Consider the following communication protocol $\pi$.
The tree of $\pi$ coincides with the tree of $F^\prime$. Inputs to $F^\prime$ will be leafs of $\pi$. In a leaf containing input variable $x_i$ the output of the protocol $\pi$ is $i$. Remaining nodes of $\pi$ are $\land$ and $\lor$ gates.
In the $\land$ gates communicates the first party, while in the $\lor$ gates communicates the second party.

Fix an  $\land$ gate $g$ (which belongs to the first party). Let $g_0, g_1$ be gates which are fed to $g$, i.e., $g = g_0 \land g_1$. There are two edges, starting at $g$, one leads to $g_0$ (and is labeled by $0$) and the other leads to $g_1$ (and is labeled by $1$). Take an input $a\in \MAJ_{2n + 1}^{-1}(0)$ to the first party. On input $a$ at the gate $g$ the first party transmits the bit $r  = \min\{ c \in \{0, 1\} \mid g_c(a) = 0\}$. If the minimum is over the empty set, then we set $r = 0$. 

Take now an $\lor$ gate $h$ belonging to the second party. Similarly, there are two edges, starting at $h$, one leads to $h_0$ (and is labeled by $0$) and the other leads to $h_1$ (and is labeled by $1$). Here $h_0, h_1$ are two gates which are fed to $h$, i.e., $h = h_0 \lor h_1$. Take an input $b\in \MAJ_{2n + 1}^{-1}(0)$ to the second party. On input $b$ at the gate $h$ the second party transmits the bit $r  = \min\{ c \in \{0, 1\} \mid h_c(\lnot b) = 1\}$.  If the minimum is over the empty set, then we set $r = 0$. Here $\lnot$ denotes the bit-wise negation. Description of the protocol $\pi$ is finished.

Clearly, the protocol $\pi$ is of depth $O(\log n)$ and its light form is polynomial-time computable.  It remains to argue that the protocol strongly computes $Q_2$-communication game for $\MAJ_{2n + 1}$. Nodes of the protocol may be identified with the gates of $F^\prime$. Consider any path $p = \langle e_1, \ldots, e_m \rangle$ in the protocol $\pi$. Assume that $e_j$ is an edge from $g^{j - 1}$ to $g^j$ and $g^0$ is the output gate of $F^\prime$. We shall show that the following: if $a\in\MAJ_{2n + 1}^{-1}(0)$ is $1$-compatible with $p$, then $g^0(a) = g^1(a) = \ldots = g^m(a) = 0$. Indeed, $g^0(a) = 0$ holds because $F^\prime$ computes $\MAJ_{2n + 1}$. Now, assume that $g^j(a) = 0$ is already proved. If $g^j$ is an $\lor$ gate, then $g^{j + 1}(a) = 0$ just because $g^{j + 1}$ feds to $g^j$. Otherwise $g^j$ is an $\land$ gate which therefore belongs to the first party. Let $r\in\{0, 1\}$ is the label of the edge $e_{j + 1}$. Note that $g^{j + 1}  = g^j_r$, where $g^j_0, g^j_1$ are two gates which are fed to $g^j$. . Since $a$ is $1$-compatible with $p$, it holds that $r$ coincides with the bit that the first party transmits at $g^j$ on input $a$, i.e., with $\min\{c \in\{0, 1\} \mid g^j_c(a) = 0\}$. The set over which the minimum is taken is non-empty because $g^j(a) = 0$. In particular $r$ belongs to this set, which means that $g^{j + 1}(a) = g^j_r(a) = 0$, as required.

Similarly one can verify that if $b\in\MAJ_{2n + 1}^{-1}(0)$ is $2$-compatible with $p$, then $g^0(\lnot b) = g^1(\lnot b) = \ldots = g^m(\lnot b) = 0$. Hence we get that  if a leaf $l$ is $1$-compatible ($2$-compatible) with $a$ ($b$) and $l$ contains a variable $x_i$, then $a_i = 0$ ($\lnot b_i = 1$). Hence the protocol strongly computes the $Q_2$-communication game for $\MAJ_{2n + 1}$.
\end{proof} 

\begin{proof}[Proof of Theorem \ref{cohen_2}]

We will realize  the protocol from the proof of Corollary \ref{weak} in such a way that it will give us  $O(\log^2 n)$-depth polynomial-size dag-like protocol with polynomial-time computable light form, strongly computing $Q_k$-communication game for $\THR^{kn + 1}_{n + 1}$. By Theorem \ref{efficient_Q_k} this means that there is a polynomial-time computable $O(\log^2 n)$-depth polynomial-size circuit $C\le \THR^{kn + 1}_{n + 1}$, consisting only of $\THR^{k + 1}_2$ gates and variables. With the same argument  as in Corollary \ref{weak} one can show that  $C$  coincides with $\THR^{kn + 1}_{n + 1}$.

We will use the same tree $T$ as in the proof of Corollary \ref{weak}. Let us specify the underlying dag $G$ of our protocol $\pi$.
 For a node $v$ of $T$ let $\mathcal{S}_v$ be the set of all tuples $(s_1, s_2, \ldots, s_k) \in \{0, 1, \ldots, kn + 1\}^k$ such that $s_1 + s_2 + \ldots + s_k < |T_v|$. For every node $v$ of $T$ and for every $(s_1, s_2, \ldots, s_k) \in \mathcal{S}_v$ the dag $G$ will contain a node identified with a tuple $(v, s_1, s_2, \ldots, s_k)$. These nodes of $G$ will be called the \emph{main nodes} (there will be some other nodes too). The starting node of $G$ will be $(r, n, \ldots, n)$, where $r$ is the root of $T$. Note that if $l$ is a leaf of $T$, then $|T_l| = 1$. Hence the only main node having $l$ as the first coordinate is $(l, 0, \ldots, 0)$. The set of terminals of $\pi$ will coincide with the set of all main nodes of the form $(l, 0, \ldots, 0)$, where $l$ is a leaf of $T$. The output of $\pi$ in $(l, 0, \ldots, 0)$ is $l$.

For an integer $s \le kn + 1$ let $W(s)$ be a binary tree of depth $O(\log n)$ with $\left|\{(a, b) \mid a, b\in \{0, 1, \ldots, s\}, a + b = s \} \right|$ leaves. We assume that leaves of $W(s)$ are identified with elements of $\{(a, b) \mid a, b\in \{0, 1, \ldots, s\}, a + b = s \}$. We use $W(s)$ in the construction of $G$. Namely, take any main node $(v, s_1, s_2, \ldots, s_k)$ with a non-leaf $v$. Attach $W(s_1)$ to it. Then attach to every leaf of $W(s_1)$ a copy of $W(s_2)$. Next, to every leaf of the resulting tree attach a copy of $W(s_3)$ and so on. In this way we obtain a binary tree $W(v, s_1, \ldots ,s_k)$ of depth $O(\log n)$ growing at $(v, s_1, \ldots, s_k)$. Its leaves can be identified with tuples of integers $(a_1, b_1, \ldots, a_k, b_k)$ satisfying $a_1, b_1, \ldots, a_k, b_k \ge 0, a_1 + b_1 = s_1, \ldots, a_k + b_k = s_k$. We will merge every leaf of $W(v, s_1, \ldots, s_k)$ with some main node. Namely, take a leaf $(a_1, b_1, \ldots, a_k, b_k)$.  If $a_1 + \ldots + a_k < |T_{v_0}|$, then we merge  $(a_1, b_1, \ldots, a_k, b_k)$ with the main node $(v_0, a_1, \ldots, a_k)$. Otherwise it should hold that $b_1 + \ldots + b_k < |T_{v_1}|$. In this case we merge $(a_1, b_1, \ldots, a_k, b_k)$ with the main node $(v_1, b_1, \ldots, b_k)$.

Description of the dag of $\pi$ is finished. Since $k$ is constant, there are $n^{O(1)}$ main nodes and to each we attach a tree of depth $O(\log n)$. Hence $\pi$ is $O(\log^2 n)$-depth and $n^{O(1)}$-size. Let us define a partition of non-terminal nodes between parties. Take a main node $(v, s_1, \ldots, s_k)$, where $v$ is not a leaf of $T$. The tree $W(v, s_1, \ldots, s_k)$, growing from $(v, s_1, \ldots, s_k)$ consists of copies of $W(s_1), \ldots, W(s_k)$.  We simply say that the $i$th party communicates in copies of $W(s_i)$.
After that we conclude that the light form of $\pi$ is polynomial-time computable.

Now let us specify how the $i$th party communicates inside $W(s_i)$. Assume that  $x\in\{0, 1\}^{kn + 1}$ is the input  to the $i$th party. If $|T_v \cap \mathrm{supp}(x)| \neq s_i$, then the $i$th party communicates  arbitrarily. Now, assume that  $|T_v \cap \mathrm{supp}(x)| = s_i$. Then the $i$th party communicates in such a way that the resulting path descends from the root of $W(s_i)$ to the leaf identified with a pair of integers $(|T_{v_0} \cap \mathrm{supp}(x)|, |T_{v_1} \cap \mathrm{supp}(x)|)$.

From this we immediately get the following observation. Let $p$ be a path from the root of $W(v, s_1, \ldots, s_k)$ to a leaf identified with a tuple $(a_1, b_1, \ldots, a_k, b_k)$. Further, assume that $x\in (\THR^{kn + 1}_{n + 1})^{-1}(0)$, satisfying $|T_v \cap \mathrm{supp}(x)| = s_i$, is $i$-compatible with $p$. Then $a_i= |T_{v_0} \cap \mathrm{supp}(x)|$ and $b_i = |T_{v_1} \cap \mathrm{supp}(x)|$. Indeed, any such $p$ passes though a copy $W(s_i)$ and leaves $W(s_i)$ in a leaf identified with  $(|T_{v_0} \cap \mathrm{supp}(x)|, |T_{v_1} \cap \mathrm{supp}(x)|)$.

From this observation one can easily deduce that if $x \in (\THR^{kn + 1}_{n + 1})^{-1}(0)$  is $i$-compatible with a  main node $(v, s_1, \ldots, s_k)$, then  $|T_v \cap \mathrm{supp}(x)| = s_i$. Indeed, we can obtain this by induction on the depth of $v$. Induction step easily follows from the previous paragraph. As for induction base we notice that  $|T_r \cap \mathrm{supp}(x)| = n$ for the root $r$ of $T$ (as in the proof of Corollary \ref{weak} we assume that $|\mathrm{supp}(x)| = n$ as party can always add missing $1$'s).

In particular, this means that $\pi$ strongly computes $Q_k$-communication game for $\THR^{kn + 1}_{n + 1}$. Indeed, any terminal of $\pi$ is of the form $(l, 0, \ldots, 0)$, where $l$ is a leaf of $T$. If  $x \in (\THR^{kn + 1}_{n + 1})^{-1}(0)$ is $i$-compatible with $(l, 0, \ldots, 0)$, then, as shown in the previous paragraph, $|T_l \cap \mathrm{supp}(x)| = |\{l\} \cap  \mathrm{supp}(x)| = 0$. This means that $x_l = 0$ and hence the output of the protocol is correct.

\end{proof}

\section{Open problems} \label{sec:open_prob}

\begin{itemize}
\item Can $Q_k$-communication game for $\THR^{kn + 1}_{n + 1}$ be solved in $O(\log n)$ bits of communication for $k\ge 3$? Equivalently, can $\THR^{kn + 1}_{n + 1}$ be computed by $O(\log n)$-depth circuit, consisting only of $\THR^{k + 1}_2$ and variables? Can a deeper look into the construction of AKS sorting network help here (note that we only use this sorting network as a black-box)?
\item Can at least $R_k$-communication game for $\THR^{kn + 1}_{n + 1}$ be solved in $O(\log n)$ bits of communication for $k\ge 3$? Again, this is equivalent to asking whether  $\THR^{kn + 1}_{n + 1}$ can be computed by $O(\log n)$-depth circuit, consisting only of $\THR^{k + 1}_2$ and \emph{literals}. Note that if we allow literals (along with $\land$ and $\lor$ gates), then there are much simpler constructions of a $O(\log n)$-depth formula for $\MAJ_n$ and, in fact, for every symmetric Boolean function~\cite{Wegener1987}. Moreover, this can be done in terms of communication complexity~\cite{BH96}. A natural approach would be to apply ideas of~\cite{BH96} to $R_k$-communication games. 
\item Are there any other interesting functions in $Q_k$ and $R_k$ which can be analyzed with our technique?
\end{itemize}

\textbf{Acknowledgments.} The authors are grateful to Alexander Shen for suggesting to generalize our initial results.

{\small
\bibliographystyle{abbrv}
\bibliography{ref}
}

\end{document}